\documentclass[12pt]{article}
\usepackage{amssymb}
\usepackage{amsmath}
\usepackage{amsthm}

\usepackage[utf8]{inputenc}

\usepackage[normalem]{ulem}

\numberwithin{equation}{section}
\newcommand{\be}{\begin{eqnarray}}
\newcommand{\ee}{\end{eqnarray}}
\newcommand{\non}{\nonumber}
\newcommand{\id}{\mathbb{I}}
\newcommand{\tr}{\mathop{\rm tr}\nolimits}

\newcommand{\eq}{\textrm{Eq.}}

\newtheorem{prop}{Proposition}% add [section] to number it within each section

\newcommand{\refs}{|0\rangle}
\newcommand{\drefs}{\langle0|}

\newcommand{\mA}{\mathcal{A}}
\newcommand{\mB}{\mathcal{B}}
\newcommand{\mC}{\mathcal{C}}
\newcommand{\mD}{\mathcal{D}}

\newcommand{\mE}{\mathcal{E}}

\newcommand{\mP}{\mathcal{P}}

\usepackage[usenames,dvipsnames]{color}

\newcommand\blfootnote[1]{%
  \begingroup
  \renewcommand\thefootnote{}\footnote{#1}%
  \addtocounter{footnote}{-1}%
  \endgroup
}

\begin{document}

\begin{titlepage}
\strut\hfill UMTG--288
\vspace{.5in}
\begin{center}

\LARGE 
Algebraic Bethe ansatz\\[0.2in] 
for the Temperley-Lieb spin-1 chain\\
\vspace{1in}
\large 
Rafael I. Nepomechie \footnote{
Physics Department,
P.O. Box 248046, University of Miami, Coral Gables, FL 33124 USA}
and Rodrigo A. Pimenta ${}^{1,}$\footnote{
Departamento de F\'{i}sica, Universidade Federal de S\~{a}o Carlos, Caixa Postal 676, 
CEP 13565-905, S\~{a}o Carlos, Brasil}\\[0.8in]
\end{center}

\vspace{.5in}

\begin{abstract}
We use the algebraic Bethe ansatz to obtain the eigenvalues and
eigenvectors of the spin-1 Temperley-Lieb open quantum chain with
``free'' boundary conditions.  We exploit the associated reflection
algebra in order to prove the off-shell equation satisfied by the
Bethe vectors.
\end{abstract}

\blfootnote{e-mail addresses: {\tt nepomechie@physics.miami.edu, pimenta@df.ufscar.br}}

\end{titlepage}

\setcounter{footnote}{0}

\section{Introduction}

The algebraic Bethe ansatz is a powerful method for solving 
both closed \cite{Faddeev:1996iy} and open \cite{Sklyanin:1988yz} integrable quantum spin 
chains.
The primary objective is to solve the spectral problem associated with
the transfer matrix.  The ultimate goals are to also compute scalar
products and correlation functions \cite{KBI}.

There are models, however, that have resisted solution by means 
of the algebraic Bethe ansatz, despite being integrable. One such 
example is the open quantum spin-$s$ chain with
``free'' boundary conditions constructed from the Temperley-Lieb (TL) 
algebra $TL_{N}$ \cite{Temperley:1971iq}, for spin $s>\frac{1}{2}$.
This is a unital algebra over the complex numbers $\mathbb{C}$ with $N-1$ generators $\{X_{(1)}, \ldots, X_{(N-1)} \}$ satisfying
\be
X^{2}_{(i)} &=& c X_{(i)} \,, \non \\
X_{(i)} X_{(i\pm 1)} X_{(i)} &=& X_{(i)}\,,  \non \\
X_{(i)} X_{(j)} &=& X_{(i)} X_{(j)}\,, \qquad |i-j| > 1 \,,
\label{TLalgebra}
\ee
where $c=-(q+q^{-1})$ and $q$ is an arbitrary parameter. The 
associated spin chain Hamiltonian is given by
\be\label{bqH}
H=\sum_{i=1}^{N-1} X_{(i)}\,.
\ee
The operator $X_{(i)}$, which acts on $\left(\mathbb{C}^{(2s+1)}\right)^{\otimes N}$, is defined by
\be
X_{(i)} = X_{i, i+1}= \id^{\otimes(i-1)} \otimes X \otimes \id^{\otimes(N-i-1)} \,, 
\ee
where $X$ is a $(2s+1)^{2}$ by $(2s+1)^{2}$ matrix with the 
following matrix elements \cite{Batchelor:1989uk}
\be
\langle m_{1}, m_{2}|X|m'_{1}, m'_{2}\rangle & = &
(-1)^{m_{1}-m'_{1}} 
Q^{m_{1}+m'_{1}}\delta_{m_{1}+m_{2},0}\delta_{m'_{1}+m'_{2},0}\,, 
\label{Xmatelems}
\ee
where $m_{1}, m_{2}, m'_{1}, m'_{2} = -s, -s+1, \ldots, s$;
and $\id$ is the identity operator on ${\mathbb C}^{2s+1}$.
The parameter $Q$ is related to $q$ by
\be
c = -(q+q^{-1}) = \left[2s+1\right]_{Q} = \frac{Q^{2s+1}-Q^{-2s-1}}{Q-Q^{-1}} =
 \sum_{k=-s}^{s} Q^{2k} \,.
\ee

The integrability of the Hamiltonian (\ref{bqH}), as well as the
possibility of solving it by algebraic Bethe ansatz, is based on the
fact that the TL algebra gives rise to solutions of the Yang-Baxter
equation by means of a procedure known as Baxterization
\cite{Jones:1990hq}.  However, the R-matrix associated with this
Hamiltonian (see Eq. (\ref{Rmatrix}) below) leads to very unusual
exchange relations for the generators of the Yang-Baxter and
reflection algebras.  This seems to be the main difficulty that has
obstructed the use of the algebraic Bethe ansatz for R-matrices from
the TL algebra.

In a previous paper \cite{Nepomechie:2016ejv}, we have proposed a number
of results related to the spectrum of (\ref{bqH}).  In particular, we
have conjectured the Bethe states and the off-shell equations that 
they satisfy.  Interestingly, such off-shell equations have a
universal character, in the sense that they are independent of the
value of the spin $s$.  The aim of this note is to present a proof of this
conjecture for $s=1$, in which case $X$ (\ref{Xmatelems}) is the 
following $9 \times 9$ matrix
\be\label{TLoperator}
X=\left(
\begin{array}{ccccccccc}
 0 & 0 & 0 & 0 & 0 & 0 & 0 & 0 & 0 \\
 0 & 0 & 0 & 0 & 0 & 0 & 0 & 0 & 0 \\
 0 & 0 & Q^{-2} & 0 & -Q^{-1} & 0 & 1 & 0 & 0 \\
 0 & 0 & 0 & 0 & 0 & 0 & 0 & 0 & 0 \\
 0 & 0 & -Q^{-1} & 0 & 1 & 0 & -Q & 0 & 0 \\
 0 & 0 & 0 & 0 & 0 & 0 & 0 & 0 & 0 \\
 0 & 0 & 1 & 0 & -Q & 0 & Q^2 & 0 & 0 \\
 0 & 0 & 0 & 0 & 0 & 0 & 0 & 0 & 0 \\
 0 & 0 & 0 & 0 & 0 & 0 & 0 & 0 & 0
\end{array}
\right)\,,
\ee
and
\be
1+Q^2+Q^{-2}=-\left(q+q^{-1}\right).
\ee
For $Q=1$, this is just the pure biquadratic spin-1 chain
\be
X\Big\vert_{Q=1}=\left( \vec S \stackrel{\cdot}{\otimes} \vec S \right)^2- \id\otimes\id \,,
\ee
while the general-$Q$ case corresponds to the $U_{Q}sl(2)$-deformation \cite{Batchelor:1989uk}.
Let us mention that the model (\ref{bqH}), as well as its closed version
with periodic boundary conditions, has been previously studied by many
authors using alternative approaches to the algebraic Bethe ansatz, see
\textit{e.g.} \cite{Parkinson1, Parkinson2, Barber:1989zz, Kluemper1, 
Kluemper2, Alcaraz:1992uq, Koberle:1993in, Kulish, Aufgebauer:2010gg}
and references therein.

We briefly review the construction of the transfer matrix 
corresponding to the Hamiltonian (\ref{bqH}) in section 
\ref{sec:transfer}. We then recall in section \ref{sec:ABA} the basic objects of the quantum inverse scattering
method and present the main results of this paper, given in 
Propositions \ref{prop1} - \ref{prop4}, which follow from a careful 
analysis of the reflection algebra. We also briefly consider
the scalar product between an on-shell Bethe vector and its off-shell dual.
We discuss our results and some further directions of investigation in section \ref{sec:discussion}. Some
functions introduced in the main text are collected in appendix \ref{sec:functions}. We present
some details of the proof of Proposition \ref{prop1} in appendix \ref{sec:proofAonB}.

\section{Transfer matrix}\label{sec:transfer}

Integrable quantum spin chains are characterized by a set of 
commuting conserved quantities (among them the Hamiltonian), whose 
generating function is the transfer matrix.
We briefly review here the construction of the transfer matrix 
corresponding to the Hamiltonian (\ref{bqH}), (\ref{TLoperator}).
The main ingredient is the R-matrix, which acts on the vector space 
$\mathbb{C}^3\otimes\mathbb{C}^3$, and is 
given (using the notation of \cite{Kulish}) by \cite{Jones:1990hq}
\be\label{Rmatrix}
R(u)=\omega(qu)\mP+\omega(u)\mP X\,,\qquad \omega(u)=u-u^{-1}\,,
\ee
where $X$ is given by (\ref{TLoperator}), and 
\be
{\cal P} = \sum_{a,b=1}^{3}e_{ab} \otimes e_{ba}\,,
\qquad \left( e_{ab} \right)_{ij} =
\delta_{a,i} \delta_{b,j} \,,
\ee
is the permutation matrix. As a consequence of the TL algebra (\ref{TLalgebra}), 
the R-matrix satisfies the Yang-Baxter equation
\be\label{YBE}
R_{12}(u/v)\, R_{13}(u)\, R_{23}(v)=R_{23}(v)\, R_{13}(u)\, R_{12}(u/v).
\ee
It also has the unitarity property
\be
R_{12}(u) R_{21}(u^{-1}) = \zeta(u)\, \id^{\otimes 2} \,,
\qquad \zeta(u) = \omega(u q^{-1})\, \omega(u^{-1} q^{-1})\,,
\label{unitarity}
\ee
where $R_{21} = {\cal P}_{12}\, R_{12}\, {\cal P}_{12} = 
R_{12}^{t_{1} t_{2}}$.
The R-matrix can be used to construct the single-row monodromy matrices
\be\label{single}
T_{0}(u)=R_{0N}(u)\dots R_{01}(u)\,,\qquad \hat{T}_0(u)=R_{10}(u)\dots R_{N0}(u)\,,
\ee
where $0$ denotes an auxiliary vector space.  It follows from the
Yang-Baxter equation that $T$ obeys the fundamental relation
\be\label{YBalgebra}
R_{12}(u/v)\, T_1(u)\, T_2(v)=T_2(v)\, T_1(u)\, R(u/v)\,,
\ee
and $\hat T$ obeys a similar relation. The double-row monodromy 
matrix \cite{Sklyanin:1988yz} is given by
\be\label{double}
U_{0}(u)=T_{0}(u)\, \hat{T}_{0}(u)\,,
\ee
and it obeys the reflection equation
\be\label{REalgebra}
R_{12}(u/v)\, U_{1}(u)\, R_{21}(uv)\, U_{2}(v)=U_{2}(v)\, 
R_{12}(uv)\, U_{1}(u)\, R_{21}(u/v)\,.
\ee
The double-row transfer matrix is given by  \cite{Sklyanin:1988yz, 
Nepomechie:2016ejv}
\be\label{t}
t(u)=\tr_{0}\left[M_{0}\, U_{0}(u)\right]\,,\qquad 
M=\textrm{diag}\left(Q^{-2},1,Q^2\right)\,.
\ee
Indeed, it has the fundamental commutativity property
\be
\left[t(u)\,, t(v)\right]=0\,,
\label{commutativity}
\ee
and it contains the Hamiltonian (\ref{bqH})
\be
H = \alpha \frac{d}{du}t(u)\Big\vert_{u=1} + \beta\, \id^{\otimes N} \,,
\label{Htransfrltn}
\ee
where
\be
\alpha = -\left[4 \omega(q^{2})\, \omega(q)^{2N-2} \right]^{-1}\,, \qquad 
\beta = \frac{\omega(q)}{\omega(q^{2})}-\frac{N}{2}\frac{\omega(q^{2})}{\omega(q)}\,.
\label{alphabeta}
\ee
The relations (\ref{commutativity}) and 
(\ref{Htransfrltn}) imply that the model (\ref{bqH}) is integrable.

\section{Algebraic Bethe ansatz}\label{sec:ABA}

We now use the algebraic Bethe ansatz to solve the spectral problem associated with
the transfer matrix (\ref{t}). Let us recall the basic needed steps: 
\begin{enumerate}
\item Identify suitable operators on the quantum space from 
auxiliary-space matrix elements of the double-row monodromy matrix (\ref{double}).
\item Identify a reference state with respect to the creation and annihilation operators.
\item Formulate convenient exchange relations from the reflection algebra (\ref{REalgebra}).
\item Define a Bethe vector as a product of creation operators acting 
on the reference state; and use the exchange relations to determine the action of the transfer matrix
on an off-shell Bethe vector.
\end{enumerate}

Let us denote the auxiliary-space matrix elements of the single-row monodromy 
matrices $T_0(u)$ and $\hat{T}_0(u)$ by 
\be
&&T_0(u)=\left(
\begin{array}{lll}
T_{11}(u) & T_{12}(u) & T_{13}(u) \\
T_{21}(u) & T_{22}(u) & T_{23}(u) \\
T_{31}(u) & T_{32}(u) & T_{33}(u)
\end{array}
\right)\,,\non\\
&&\hat{T}_0(u)=\left(
\begin{array}{lll}
\hat{T}_{11}(u) & \hat{T}_{12}(u) & \hat{T}_{13}(u) \\
\hat{T}_{21}(u) & \hat{T}_{22}(u) & \hat{T}_{23}(u) \\
\hat{T}_{31}(u) & \hat{T}_{32}(u) & \hat{T}_{33}(u)
\end{array}
\right),
\label{singlemono}
\ee
where each entry acts on the quantum space $\left({\mathbb C}^{3}\right)^{\otimes N}$.
It is convenient to denote the auxiliary-space matrix elements of the double-row monodromy matrix (\ref{double})
by
\be\label{mono}
U_0(u)=\left(
\begin{array}{ccc}
\mA(u) & \mB_1(u) & \mB(u) \\
\mC_1(u) & \mE(u)+\mA(u) & \mB_2(u) \\
\mC(u) & \mC_2(u) & \mD(u)+y(u)\mA(u)
\end{array}
\right)\,,
\ee
where
\be\label{ddef}
y(u)=1-Q^{-2}d(u)\,,\qquad d(u)=-\frac{\omega(u^2)}{\omega(qu^2)}\,.
\ee
Three-dimensional representations such as (\ref{singlemono})
have been used to solve periodic 19-vertex \cite{Tarasov} and nested 15-vertex 
\cite{Kulish:1983rd} models; and (\ref{mono}) is a generalization
for the open case, see for instance \cite{Fan1997409} and \cite{FOERSTER1993512}.

The operator entries of the double-row monodromy matrix 
(\ref{mono}) are given in terms of single-row monodromy 
matrix elements $T_{ij}$ and $\hat{T}_{ij}$ by means of (\ref{double}). In terms
of the double-row operators, the transfer matrix (\ref{t}) can be written as
\be\label{transfer}
t(u)=a(u)\,\mA(u)+Q^2\,\mD(u)+\mE(u) \,,
\ee
where
\be\label{adef}
a(u)=-\frac{\omega(q^2u^2)}{\omega(qu^2)}\,.
\ee
\subsection{Bethe vector}
Having defined the operator representation, we need to find a convenient reference state. We note that
\be\label{refstate}
\refs=\left( 
\begin{array}{c}
1 \\ 
0 \\
0
\end{array}
\right)^{\otimes N}
\ee
satisfies the following properties
\be\label{singlerep}
&&T_{ij}(u)\refs=\hat T_{ij}(u)\refs=0\,,\qquad \textrm{for}\qquad i>j\,,\non\\
&&T_{11}(u)\refs=\hat T_{11}(u)\refs=\omega(qu)^N\refs\,,\non\\
&&T_{22}(u)\refs= \hat T_{22}(u)\refs=0\,,\non\\
&&T_{33}(u)\refs=\hat T_{33}(u)\refs=\omega(u)^N\refs\,.
\ee
Moreover, the fundamental relation (\ref{YBalgebra}) evaluated at the point $u=v^{-1}$, taking into
account (\ref{unitarity}) and (\ref{singlerep}), gives
\be\label{tijtji}
&&T_{21}(u)\, \hat T_{12}(u)\refs=T_{11}(u)\, \hat T_{11}(u)\refs\,,\non\\
&&T_{32}(u)\, \hat T_{23}(u)\refs=Q^{-2}T_{33}(u)\, \hat T_{33}(u)\refs\,,\non\\
&&T_{31}(u)\, \hat T_{13}(u)\refs=y(u)\, T_{11}(u)\, \hat 
T_{11}(u)\refs-\left(Q^{-2}+y(u)\right)T_{33}(u)\, \hat T_{33}(u)\refs\,\,,\non\\
&&T_{11}(u)\, \hat T_{12}(u)\refs=0\,.
\ee
The results (\ref{singlerep}) and (\ref{tijtji}) imply
\be\label{doublerep}
&&\mA(u)\refs=\Lambda_1(u)\refs\,,\non\\
&&\mD(u)\refs=Q^{-2}\,d(u)\,\Lambda_2(u)\refs\,,\non\\
&&\mE(u)\refs=0\,,
\ee
where
\be\label{Lambda12}
\Lambda_1(u)=\omega(qu)^{2N}\,,\qquad \Lambda_2(u)=\omega(u)^{2N}\,.
\ee
In addition, we also have the following properties for the off-diagonal
double-row operators,
\be\label{caction}
&&\mC(u)\refs=\mC_1(u)\refs=\mC_2(u)\refs=0\,,
\ee
and
\be\label{b1action}
&&\mB_1(u)\refs=0\,.
\ee
Due to (\ref{doublerep}), we see that the reference 
state $\refs$ is an eigenstate of the transfer matrix,
\be
&&t(u)\refs=\left(a(u)\Lambda_1(u)+d(u)\Lambda_2(u)\right)\refs.
\ee

Since the operator $\mB_1(u)$ (in addition to $\mC(u)$, $\mC_1(u)$ and $\mC_2(u)$) annihilates the reference state, we are left in principle with two operators to play the role of raising operators, either $\mB(u)$ or $\mB_2(u)$.
However, the reflection algebra (\ref{REalgebra}) strongly suggests that  $\mB(u)$ is the correct
choice. Indeed, after some manipulation\footnote{Let us call $\textrm{eq}[i,j]$ the $(i,j)$ entry of  
equation (\ref{REalgebra}) regarded 
as a $9\times 9$ matrix in the auxiliary space. 
The exchange relations (\ref{AB}) - (\ref{BE}) all follow from
$\textrm{eq}[1,3]$, $\textrm{eq}[1,7]$, $\textrm{eq}[1,8]$,
$\textrm{eq}[1,9]$, $\textrm{eq}[4,6]$ and $\textrm{eq}[7,9]$.}, we found the following
exchange relations from the reflection algebra,
\be\label{AB}
&&\mA(u)\mB(v)=f(u,v)\mB(v)\mA(u)+f_1(u,v)\mB(u)\mA(v)+f_2(u,v)\mB(u)\mD(v)\non\\&&\qquad\qquad\qquad+
f_3(u,v)\mB(u)\mE(v)-\mB_1(u)\mB_2(v)\,,
\ee
\be\label{DB}
&&\mD(u)\mB(v)=h(u,v)\mB(v)\mD(u)+h_1(u,v)\mB(u)\mD(v)+h_2(u,v)\mB(u)\mA(v)\non\\&&\qquad\qquad\qquad+
h_3(u,v)\mB(u)\mE(v)+Q^{-2}a(u)\mB_1(u)\mB_2(v)-Q^{-2}\mE(u)\mB(v)\,,
\ee
\be\label{BB}
&&\mB(u)\mB(v)=\mB(v)\mB(u)\,,
\ee
\be\label{BB1}
&&\mB(u)\mB_1(v)=\mB(v)\mB_1(u)\,,
\ee
\be\label{BE}
&&\mB(u)\mE(v)=\mB(v)\mE(u)\,,
\ee
where the coefficients are given by
\be\label{coeffj}
&&f(u,v)=\frac{\omega(uq^{-1}v^{-1})\omega(uv)}{\omega(uv^{-1})\omega(quv)}\,,\non\\
&&f_1(u,v)=\frac{\omega(v^2)}{\omega(qv^2)\omega(uv^{-1})}\left(\omega(qvu^{-1})+Q^{-2}\omega(vu^{-1})\right)\,,\non\\
&&f_2(u,v)=-1-\frac{Q^2\omega(uv)}{\omega(quv)}\,,\non\\
&&f_3(u,v)=-\frac{\omega(uv)}{\omega(quv)}\,,\qquad
\ee
and
\be\label{coefhj}
&&h(u,v)=\frac{\omega(uqv^{-1})\omega(q^2uv)}{\omega(uv^{-1})\omega(quv)}\,,\non\\
&&h_1(u,v)=\frac{\omega(q^2u^2)}{Q^2\omega(qu^2)\omega(uv^{-1})}\left(\left(1+Q^{-2}\right)\omega(quv^{-1})+\omega(q^2uv^{-1})\right)\,,\non\\
&&h_2(u,v)=\frac{\omega(q^2u^2)\omega(v^2)}{Q^4\omega(qu^2)\omega(quv)\omega(qv^2)}\left(\left(1+Q^{-2}\right)\omega(q^2uv)+\omega(q^3uv)\right)\,,\non\\
&&h_3(u,v)=\frac{\left(q-q^{-1}\right)\omega(quv^{-1})}{Q^2\omega(qu^2)\omega(quv)}\,.
\ee
We can observe that the commutation relations (\ref{AB}) and (\ref{DB})
have a structure similar to those for the six-vertex model \cite{Sklyanin:1988yz},
although with some extra terms. The relation (\ref{BB}) guarantees that the vector
$\mB(u_1)\dots\mB(u_M)\refs$ is a symmetric quantity in its arguments.
Thus, the operator $\mB(u)$ is indeed a good raising operator candidate.
The relations (\ref{BB1}) and (\ref{BE}) do not have an analogous
counterpart in the six-vertex model, but they play a fundamental role in the
algebraic Bethe ansatz analysis for the present case.

At this point, it is convenient to introduce a shorthand notation, as follows:
\begin{itemize}
\item For a set of $M$ rapidities, we will use the notation $\bar 
u=\left\{u_1,\dots,u_M\right\}$, where the cardinality 
of $\bar u$ is $\#\bar u=M$.
\item If the $i$th rapidity is dropped from the set $\bar u$, we denote $\bar u_i=\left\{u_1,\dots,u_{i-1},u_{i+1},\dots,u_M\right\}$.
\end{itemize}
In addition, let us introduce the following strings of operators,
\be
&&B^M(\bar u)=\prod_{i=1}^M\mB(u_i)\,,\non\\
&&B_{i}^M(\left\{u,\bar u_i\right\})=\mB(u)\prod_{j\neq i}^M\mB(u_j)\,,\non\\
&&\bar B_{i}^M(\left\{u,\bar u\right\})=\prod_{j=0}^{i-1}\mB(u_j)\mB_1(u_i)\mB_2(u_{i+1})\prod_{j=i+2}^M\mB(u_j)\,,\non\\
&&\tilde B_{i}^M(\left\{u,\bar u\right\})=\prod_{j=0}^{i-1}\mB(u_j)\mE(u_i)\prod_{j=i+1}^M\mB(u_j)\,,
\ee
for $\#\bar u=M$ and where we identify $u_0\equiv u$. Let us also introduce the vectors
\be\label{BV}
|\bar u\rangle=B^M(\bar u)\refs\,,
\ee
and
\be\label{auxBV}
|\left\{u,\bar u_i\right\}\rangle=B_{i}^M(\left\{u,\bar u_i\right\})\refs\,.
\ee

Following the previous discussion, we propose that  
the Bethe vectors are given by (\ref{BV}).
We therefore need to compute the action of the transfer matrix 
(\ref{transfer}) on this vector.
The result is obtained as a consequence of the following proposition
\begin{prop}\label{prop1}
The action of the operator $\mA(u)$ on the string $B^M(\bar u)$, with $\#\bar u=M$, is given by,
\be\label{AonBstring}
&&\mA(u)B^M(\bar u)=B^M(\bar 
u)\mA(u)\prod_{i=1}^Mf(u,u_i)+\sum_{i=1}^M\hat F_i^M(\left\{u,\bar 
u\right\})+\sum_{i=2}^M\hat Z_i^M(\left\{u,\bar u\right\})\non\\
&&\qquad+
\sum_{i=0}^{M-1}r_i\bar B_i^M(\left\{u,\bar u\right\})
+\sum_{i=1}^{M-1}s_i\tilde B_i^M(\left\{u,\bar u\right\})+\alpha^{M}(\left\{u,\bar u\right\})\tilde B_M^M(\left\{u,\bar u\right\})\,,
\ee
while the action of the operator $\mD(u)$ on the string $B^M(\bar u)$ is given by,
\be\label{DonBstring}
&&\mD(u)B^M(\bar u)=B^M(\bar u)\mD(u)\prod_{i=1}^Mh(u,u_i)+\sum_{i=1}^M\hat G_i^M(\left\{u,\bar u\right\})-Q^{-2}a(u)\sum_{i=2}^M\hat Z_i^M(\left\{u,\bar u\right\})
\non\\
&&\qquad-	
Q^{-2}a(u)\sum_{i=0}^{M-1}r_i\bar B_i^M(\left\{u,\bar u\right\})
-Q^{-2}a(u)\sum_{i=1}^{M-1}s_i\tilde B_i^M(\left\{u,\bar u\right\})+\delta^{M}(\left\{u,\bar u\right\})\tilde B_M^M(\left\{u,\bar u\right\})\non\\&&
\qquad-Q^{-2}\mE(u)B^M(\bar u)\,,
\ee
where
\be
&&\hat F_i^M(\left\{u,\bar u\right\})=B_i^{M}(\left\{u,\bar u_i\right\})\mA(u_i)f_1(u,u_i)\prod_{j\neq i}^Mf(u_i,u_j)
\non\\&&\qquad+B_i^{M}(\left\{u,\bar u_i\right\})\mD(u_i)f_2(u,u_i)\prod_{j\neq i}^Mh(u_i,u_j)
\ee
and
\be
&&\hat G_i^M(\left\{u,\bar u\right\})=B_i^{M}(\left\{u,\bar u_i\right\})\mA(u_i)h_2(u,u_i)\prod_{j\neq i}^Mf(u_i,u_j)
\non\\&&\qquad+B_i^{M}(\left\{u,\bar u_i\right\})\mD(u_i)h_1(u,u_i)\prod_{j\neq i}^Mh(u_i,u_j)\,.
\ee
The explicit expressions for the
functions $\hat Z_i^ M$, $r_i$, $s_i$, $\alpha^{M}$ and $\delta^{M}$ are not relevant for the following and are collected in appendix \ref{sec:functions}.
\end{prop}
\begin{proof}
The proof is obtained by induction on $M$ and by the use of the commutation relations (\ref{AB}), (\ref{DB}),
(\ref{BB}), (\ref{BB1}) and (\ref{BE}). We provide some details in appendix \ref{sec:proofAonB}.
\end{proof}

We note an intricate structure of the operator relations (\ref{AonBstring}) and (\ref{DonBstring}),
with many ``unwanted'' terms.
However, the final off-shell equation satisfied by the Bethe vector is amazingly simple, since many of terms from (\ref{AonBstring})
cancel with those from (\ref{DonBstring}). Indeed, we can now easily prove the following proposition
\begin{prop}\label{prop2}
The off-shell equation for the transfer matrix (\ref{transfer}) acting on the Bethe vector (\ref{BV}) is given by
\be\label{offshell}
t(u)|\bar u\rangle=\Lambda(\left\{u,\bar u\right\})|\bar u\rangle+\sum_{i=1}^MH(u,u_i)E(\left\{u_i,\bar u_i\right\})|\left\{u,\bar u_i\right\}\rangle\,,
\ee
where
\be\label{Lambda}
\Lambda(\left\{u,\bar u\right\})=a(u)\Lambda_1(u)\prod_{i=1}^Mf(u,u_i)+d(u)\Lambda_2(u)\prod_{i=1}^Mh(u,u_i)\,,
\ee
\be\label{BEq}
E(\left\{u_i,\bar u_i\right\})=\Lambda_1(u_i)\prod_{j\neq i}^Mf(u_i,u_j)-\Lambda_2(u_i)\prod_{j\neq i}^Mh(u_i,u_j)\,,
\ee
and
\be\label{Hdef}
H(u,v)=(q-q^{-1})\frac{\omega(q^2u^2)}{\omega(uv^{-1})\omega(quv)}d(v)\,.
\ee
The functions entering equations (\ref{Lambda}) and (\ref{BEq}) are given by (\ref{ddef}), (\ref{adef}), (\ref{Lambda12}), (\ref{coeffj}) and
(\ref{coefhj}). 
\end{prop}
\begin{proof}
Applying (\ref{AonBstring}) and (\ref{DonBstring}) on the reference state (\ref{refstate}), taking into account
(\ref{doublerep}), (\ref{BV}) and (\ref{auxBV}), and the identities,
\be
&&a(u)f_1(u,v)+Q^2h_2(u,v)=H(u,v)\,,\non\\&& a(u)f_2(u,v)+Q^2h_1(u,v)=-\frac{Q^2}{d(v)}H(u,v)\,,
\ee
we obtain (\ref{offshell}). Note that the terms with 
$\alpha^M$ and $\delta^M$ separately vanish when acting on the 
reference state since $\tilde B_M^M=\prod_{j=0}^{M-1}\mB(u_j)\mE(u_M)$, and $\mE$ annihilates the reference state.
\end{proof}

The result (\ref{offshell}) has been conjectured for 
arbitrary values of spin
in \cite{Nepomechie:2016ejv}\footnote{We use here
a slightly different notation compared with \cite{Nepomechie:2016ejv}. In particular, the functions
$a(u)$ and $d(u)$ are related as follows: 
$a_{\textrm{previous}}(u)=a_{\textrm{now}}(u)\Lambda_1(u)$ and $d_{\textrm{previous}}(u)=d_{\textrm{now}}(u)\Lambda_2(u)$.}. By imposing $E(\left\{u_i,\bar u_i\right\})=0$ for $i=1,\dots M$ (the Bethe equations), we obtain the eigenvalues (\ref{Lambda}) and the eigenvectors (\ref{BV})
of the transfer matrix (\ref{transfer}). It is remarkable that the 
off-shell equation (\ref{offshell}) has exactly the same form 
as the analogous relation
for the XXZ spin-$\frac{1}{2}$ chain with $U_qsl(2)$ symmetry \cite{Kulish:1991np}, see \textit{e.g.} equation (A.17) of \cite{Gainutdinov:2015vba}.

\subsection{Dual Bethe vector}
We can follow a similar procedure to obtain the dual Bethe vectors.
Indeed, defining
\be
\drefs={\left(
\begin{array}{cccc}1 & 0 & \cdots & 0
\end{array}\right)}^{\otimes N}
\ee
such that $\drefs0\rangle=1$,
we can obtain, as before,
\be\label{ddoublerep}
&&\drefs\mA(u)=\drefs\Lambda_1(u)\,,\non\\
&&\drefs\mD(u)=\drefs Q^{-2}d(u)\Lambda_2(u)\,,\non\\
&&\drefs\mE(u)=0\,,
\ee
where $\Lambda_1(u)$ and $\Lambda_2(u)$ are given by (\ref{Lambda12}). We also have
 \be\label{bcaction}
&&\drefs\mB(u)=\drefs\mB_1(u)=\drefs\mB_2(u)=\drefs\mC_1(u)=0\,.
\ee
For this case,
the needed commutation relations are \footnote{The commutation relations
now follow from $\textrm{eq}[3,1]$, $\textrm{eq}[7,1]$,  $\textrm{eq}[8,1]$, $\textrm{eq}[9,1]$, $\textrm{eq}[6,4]$ and $\textrm{eq}[9,7]$.}
\be\label{CA}
&&\mC(v)\mA(u)=f(u,v)\mA(u)\mC(v)+f_1(u,v)\mA(v)\mC(u)+f_2(u,v)\mD(v)\mC(u)\non\\&&\qquad\qquad\qquad+
f_3(u,v)\mE(v)\mC(u)-\mC_2(v)\mC_1(u)\,,
\ee
\be\label{CD}
&&\mC(v)\mD(u)=h(u,v)\mD(u)\mC(v)+h_1(u,v)\mD(v)\mC(u)+h_2(u,v)\mA(v)\mC(u)\non\\&&\qquad\qquad\qquad+
h_3(u,v)\mE(v)\mC(u)+Q^{-2}a(u)\mC_2(v)\mC_1(u)-Q^{-2}\mC(v)\mE(u)\,,
\ee
\be\label{CC}
&&\mC(v)\mC(u)=\mC(u)\mC(v)\,,
\ee
\be\label{C1C}
&&\mC_1(v)\mC(u)=\mC_1(u)\mC(v)\,,
\ee
\be\label{EC}
&&\mE(v)\mC(u)=\mE(u)\mC(v)\,.
\ee
Let us also introduce
\be
&&C^M(\bar u)=\prod_{i=M}^1\mC(u_i)\,,\non\\
&&C_{i}^M(\left\{u,\bar u_i\right\})=\prod_{j=M,j\neq i}^1\mC(u_j)\mC(u)\,,\non\\
&&\bar C_{i}^M(\left\{u,\bar u\right\})=\prod_{j=M}^{i+2}\mC(u_j)\mC_2(u_{i+1})\mC_1(u_{i})\prod_{j=i-1}^0\mC(u_j)\,,\non\\
&&\tilde C_{i}^M(\left\{u,\bar u\right\})=\prod_{j=M}^{i+1}\mC(u_j)\mE(u_i)\prod_{j=i-1}^0\mC(u_j)\,,
\ee
for $\#\bar u=M$ and where we identify again $u_0\equiv u$. In the above definitions, the product indices run
backwards\footnote{Throughout this subsection, we use the following ordering of the rapidities: $\left\{u_M,\dots,u_1\right\}$. While for the
products of operators $C^M$ and $C_i^M$ this ordering is irrelevant 
thanks to (\ref{CC}), it is important to
maintain this ordering for
the auxiliary products $\bar C_i^M$ and $\tilde C_i^M$, when compared with $\bar B_i^M$ and $\tilde B_i^M$.}. Let us define
\be\label{dBV}
\langle \bar u| = \drefs C^M(\bar u)\,,
\ee
and
\be\label{auxdBV}
\langle \left\{u,\bar u_i\right\}|= \drefs C_{i}^M(\left\{u,\bar u_i\right\})\,.
\ee
We use (\ref{dBV}) as the dual Bethe vector, and
compute the action of the transfer matrix (\ref{transfer}) on it.
The result is
\begin{prop}\label{prop3}
The (left) action of the operator $\mA(u)$ on the string $C^M(\bar u)$, with $\#\bar u=M$, is given by
\be\label{AonCstring}
&&C^M(\bar u)\mA(u)=\mA(u)C^M(\bar u)\prod_{i=1}^Mf(u,u_i)+\sum_{i=1}^M\check F_i^M(\left\{u,\bar u\right\})+\sum_{i=2}^M\check Z_i^M(\left\{u,\bar u\right\})\non\\
&&\qquad+
\sum_{i=0}^{M-1}r_i\bar C_i^M(\left\{u,\bar u\right\})
+\sum_{i=1}^{M-1}s_i\tilde C_i^M(\left\{u,\bar 
u\right\})+\alpha^{M}(\left\{u,\bar u\right\})\tilde 
C_M^M(\left\{u,\bar u\right\}) \,,
\ee
while the (left) action of the operator $\mD(u)$ on the string $C^M(\bar u)$ is given by
\be\label{DonCstring}
&&C^M(\bar u)\mD(u)=\mD(u)C^M(\bar u)\prod_{i=1}^Mh(u,u_i)+\sum_{i=1}^M\check G_i^M(\left\{u,\bar u\right\})-Q^{-2}a(u)\sum_{i=2}^M\check Z_i^M(\left\{u,\bar u\right\})
\non\\
&&\qquad-	
Q^{-2}a(u)\sum_{i=0}^{M-1}r_i\bar C_i^M(\left\{u,\bar u\right\})
-Q^{-2}a(u)\sum_{i=1}^{M-1}s_i\tilde C_i^M(\left\{u,\bar u\right\})+\delta^{M}(\left\{u,\bar u\right\})\tilde C_M^M(\left\{u,\bar u\right\})\non\\&&
\qquad-Q^{-2}C^M(\bar u)\mE(u) \,,
\ee
where
\be
&&\check F_i^M(\left\{u,\bar u\right\})=\mA(u_i)C_i^{M}(\left\{u,\bar u_i\right\})f_1(u,u_i)\prod_{j\neq i}^Mf(u_i,u_j)
\non\\&&\qquad+\mD(u_i)C_i^{M}(\left\{u,\bar 
u_i\right\})f_2(u,u_i)\prod_{j\neq i}^Mh(u_i,u_j) \,,
\ee
and
\be
&&\check G_i^M(\left\{u,\bar u\right\})=\mA(u_i)C_i^{M}(\left\{u,\bar u_i\right\})h_2(u,u_i)\prod_{j\neq i}^Mf(u_i,u_j)
\non\\&&\qquad+\mD(u_i)C_i^{M}(\left\{u,\bar u_i\right\})h_1(u,u_i)\prod_{j\neq i}^Mh(u_i,u_j)\,.
\ee
The explicit expressions for the
functions $\check Z_i^ M$, $r_i$, $s_i$, $\alpha^{M}$ and $\delta^{M}$ are not relevant for the following and are given in appendix \ref{sec:functions}.
\end{prop}
\begin{proof}
The proof is similar to the proof of Proposition \ref{prop1}; the needed commutation relations are now (\ref{CA}), (\ref{CD}), (\ref{CC}), (\ref{C1C}) and
(\ref{EC}).
\end{proof}
Using the previous proposition, we obtain
\begin{prop}\label{prop4}
The (left) off-shell equation for the transfer matrix (\ref{transfer}) acting on the Bethe vector (\ref{dBV}) is given by
\be\label{leftoffshell}
\langle \bar u|t(u)=\langle \bar u|\Lambda(\left\{u,\bar u\right\})+\sum_{i=1}^M\langle \left\{u,\bar u_i\right\}|H(u,u_i)E(\left\{u_i,\bar u_i\right\})\,,
\ee
with the same functions as in (\ref{offshell}).
\end{prop}
\begin{proof}
The proof is similar to the proof of Proposition \ref{prop2}.
\end{proof}
Again, the result (\ref{leftoffshell}) has been conjectured
for arbitrary values of spin
in \cite{Nepomechie:2016ejv}. The (left) eigenvalues and eigenvectors
of the transfer matrix are obtained by imposing $E(\left\{u_i,\bar u_i\right\})=0$ for $i=1,\dots,M$.

\subsection{Scalar product}
Let us now briefly consider the scalar product between the Bethe vector (\ref{BV}) and the dual Bethe vector (\ref{dBV}). 
In the paper \cite{Nepomechie:2016ejv},
we have proposed that the scalar product
between an on-shell state $\langle \bar u|$ and an arbitrary off-shell state 
$|\bar v\rangle$ is given by
\be
\langle \bar u|\bar v\rangle
=\left(\frac{1}{2Q^{2s}}\right)^M\,
\prod_{i=1}^M\frac{\omega(u_i)^{2N}u_i\,\omega(u_i^2)}{\omega(u_i^2q)\omega(v_i^2q^2)}
\prod_{j<i}^M\frac{\omega(u_iu_jq^2)}{\omega(u_iu_j)}
\frac{\textrm{Det}_M\left(\frac{\partial}{\partial u_i}\Lambda(\left\{v_j,\bar u\right\})\right)}
{\textrm{Det}_M\left(\frac{1}{\omega(v_iu_j^{-1})\omega(v_iu_jq)}\right)}\,,
\label{slavnov}
\ee
where $\# \bar u=\# \bar v=M$ and the set $\bar u$ is a solution of the Bethe equations, \textit{i.e.}, $E(\left\{u_i,\bar u_i\right\})=0$
for $i=1,\dots,M$. Here, we have $s=1$. A formula of the type (\ref{slavnov})
is generally known as a Slavnov formula \cite{Sla89}, while its limit $v_k\rightarrow u_k$
(the square of the norm) is known as a Gaudin-Korepin formula \cite{Gaudin,PhysRevD.23.417,Korepin:1982gg}. 
For the $s=\frac{1}{2}$ chain with (diagonal) open boundary
conditions, the formula (\ref{slavnov}) was obtained in \cite{Kitanine:2007bi}
(see also \cite{Wang2002633} for the XXX chain), using a method 
different from the one in \cite{Sla89}.

The proof of the formula (\ref{slavnov}) remains an open problem for 
$s>\frac{1}{2}$. We now
briefly comment on the obstacles that we have encountered for $s=1$, which come up already
in the simplest $M=1$ case. The scalar product (\ref{slavnov}) for $M=1$ can be in principle obtained from the commutation
relation between the operators $\mC(u_1)$ and $\mB(v_1)$, which is given by\footnote{This commutation relation follows from $\textrm{eq}[7,3]$ of the the reflection algebra.}
\be
\mC(u_1)\mB(v_1)&=&\mB(v_1)\mC(u_1)\non\\
&&+x_1(u_1,v_1)\mA(u_1)\mA(v_1)+x_2(u_1,v_1)\mA(v_1)\mA(u_1)+x_3(u_1,v_1)\mD(u_1)\mA(v_1)\non\\
&&+x_4(u_1,v_1)\mA(u_1)\mD(v_1)+x_5(u_1,v_1)\mA(v_1)\mD(u_1)+x_6(u_1,v_1)\mD(u_1)\mD(v_1)\non\\
&&+y_1(u_1,v_1)\mA(v_1)\mE(u_1)+y_2(u_1,v_1)\mE(u_1)\mA(v_1)+y_3(u_1,v_1)\mE(u_1)\mD(v_1)\non\\
&&+\mB_1(v_1)\mC_1(u_1)-\mC_2(u_1)\mB_2(v_1)\,,
\label{CB}
\ee
where the coefficients are given in appendix \ref{sec:functions}.
Note that the first three lines of (\ref{CB}) are similar to the analogous relation in the six-vertex model; all the other terms are new. Applying (\ref{CB}) on the reference
state $\refs$, taking into account (\ref{doublerep}) and (\ref{caction}),  and projecting the result on $\drefs$, we obtain
\be
\langle u_1|v_1\rangle &=&
\left[x_1(u_1,v_1)+
x_2(u_1,v_1)\right]\Lambda_1(u_1)\Lambda_1(v_1)-\frac{
\omega(v_1^2)x_4(u_1,v_1)}{Q^2\omega(qv_1^2)}\Lambda_1(u_1)\Lambda_2(v_1)\non\\&&-\frac{\omega(u_1^2)\left[x_3(u_1,v_1)+x_5(u_1,v_1)\right]}{Q^2\omega(qu_1^2)}\Lambda_2(u_1)\Lambda_1(v_1)
+\frac{\omega(u_1^2)\omega(v_1^2)x_6(u_1,v_1)}{Q^4\omega(qu_1^2)\omega(qv_1^2)}\Lambda_2(u_1)\Lambda_2(v_1)
\non\\&&-\drefs\mC_2(u_1)\mB_2(v_1)\refs\,.
\label{CBref}
\ee
where we observe that most of the extra terms in (\ref{CB}) do not contribute to the
scalar product, except for $\mC_2(u_1)\mB_2(v_1)$. We now suppose
that the variable $u_1$ is a Bethe root. Under this condition, we have numerically checked (up to $N=6$) that
\be
\drefs\mC_2(u_1)=0\,.
\ee
Then, using in (\ref{CBref}) the fact that $\Lambda_1(u_1)=\Lambda_2(u_1)$ when $u_1$ is a Bethe root,
we obtain by explicit computation the right-hand side of (\ref{slavnov}) for $M=1$. For $M>1$, new terms (when compared to
the analogous relations in the six-vertex model) appear in the off-shell/off-shell scalar
product $\langle \bar u|\bar v\rangle$. All these terms, however, presumably disappear
when $\langle \bar u|$ is on-shell; the proof of this fact, which 
would be a first step towards proving the formula (\ref{slavnov}), has so far eluded us.

\section{Discussion}\label{sec:discussion}

We have considered the quantum spin-1 chain with
``free'' boundary conditions constructed from the TL algebra in the algebraic Bethe ansatz framework.
The main result of this note is the proof of the off-shell equations
satisfied by the Bethe vector, see Proposition \ref{prop2}, and by the dual Bethe vector, see Proposition \ref{prop4}.
The complexity of the proof originates from the unusual exchange relations (\ref{AB}) - (\ref{BE}); 
and we believe it is quite remarkable that they lead to such simple 
off-shell equations (\ref{offshell}). 

We note that despite of the fact that the auxiliary space is 3-dimensional, the off-shell equations
have the same form of those of the quantum-group-invariant XXZ spin-$\frac{1}{2}$ chain (which has a 2-dimensional
auxiliary space). This is a step towards a proof of the more general 
conjecture in \cite{Nepomechie:2016ejv}, which states that the off-shell
equation of TL spin chains with ``free'' boundary conditions 
associated with the spin-$s$ representation of $U_{Q}sl(2)$ is actually
universal, \textit{i.e.}, it is independent of the value of the spin. We hope that the results presented here can be
further developed in order to prove the formula (\ref{slavnov}) for the scalar product between the off-shell Bethe vector and its on-shell dual, 
which is also independent (up to a constant factor) of the value of 
the spin.

According to conventional wisdom and experience, closed chains should 
be simpler than corresponding open chains. However, we have seen that 
this is not the case for the TL model. Nevertheless,
the method presented here may also shed some light on the algebraic Bethe ansatz formulation for the
closed TL spin-1 chain with periodic boundary conditions. Interestingly,
the Yang-Baxter algebra (\ref{YBalgebra}) seems subtler than the associated reflection algebra (\ref{REalgebra}).
Indeed, the  the spectrum of the closed chain is characterized by a 
``dynamically'' generated twist, see \cite{Alcaraz:1992uq, 
Aufgebauer:2010gg, Finch:2014nxa, Finch:2015,Nepomechie:2016ejv}.
It would be interesting to obtain the Bethe vectors and the associated Bethe
equations from the Yang-Baxter algebra.

As a further direction of investigation, it may be worth to consider the algebraic Bethe ansatz formulation of the spin-1 TL chain with more
complicated boundary interactions. Both diagonal and non-diagonal reflection matrices are available \cite{LimaSantos:2010nw, Avan:2010mh}. For the former,
the procedure described in this paper can
probably be applied without 
significant changes.
For the latter, one would have to extend the modified algebraic Bethe ansatz, 
see \cite{Belliard:2014fsa,Belliard:2014rna,Avan:2015ada} for the XXZ spin-$\frac{1}{2}$ chain with
non-diagonal boundaries, or to extend the construction of the on-shell Bethe states
from the off-diagonal Bethe ansatz \cite{Wang2015}.

\section*{Acknowledgments}
The work of RN was supported in part by the National Science
Foundation under Grant PHY-1212337, and by a Cooper fellowship.
RP thanks the S\~ao Paulo Research Foundation (FAPESP),
grants \# 2014/00453-8 and \# 2014/20364-0, for financial support.
We also acknowledge the support by FAPESP and the University of Miami under the SPRINT grant \#2016/50023-5.

\appendix

\section{Functions}\label{sec:functions}

We list here some functions used in the main text
\be
&&\hat Z_i^M(\left\{u,\bar u\right\})=(q+q^{-1})
\sum_{k=2}^i
r_{k-2}
\left(-Q^{-2}d(u_i)B_i^{M}(\left\{u,\bar u_i\right\})\mA(u_i)\prod_{j=k,j\neq i}^Mf(u_i,u_j)\right.\non\\
&&\qquad+
\left.B_i^{M}(\left\{u,\bar u_i\right\})\mD(u_i)\prod_{j=k,j\neq i}^Mh(u_i,u_j)\right)\,,
\ee
\be
&&\check Z_i^M(\left\{u,\bar u\right\})=(q+q^{-1})
\sum_{k=2}^i
r_{k-2}
\left(-Q^{-2}d(u_i)\mA(u_i)C_i^{M}(\left\{u,\bar u_i\right\})\prod_{j=k,j\neq i}^Mf(u_i,u_j)\right.\non\\
&&\qquad+
\left.\mD(u_i)C_i^{M}(\left\{u,\bar u_i\right\})\prod_{j=k,j\neq i}^Mh(u_i,u_j)\right)\,,
\ee
\be
&&\alpha^{M}(\left\{u,\bar u\right\})=f_3(u,u_M)\prod_{i=1}^{M-1}f(u,u_i)\non\\&&\qquad+
\sum_{i=1}^{M-1}\left\{f_1(u,u_i)f_3(u_i,u_M)\prod_{j\neq i}^{M-1}f(u_i,u_j)+
f_2(u,u_i)h_3(u_i,u_M)\prod_{j\neq i}^{M-1}h(u_i,u_j)\right.\non\\
&&\qquad\qquad+(q+q^{-1})\sum_{k=2}^ir_{k-2}\left(-Q^{-2}d(u_i)f_3(u_i,u_M)\prod_{j=k,j\neq i}^{M-1}f(u_i,u_j)\right.\non\\
&&\qquad\qquad\qquad+\left.\left.
h_3(u_i,u_M)\prod_{j=k,j\neq i}^{M-1}h(u_i,u_j)\right)\right\}\,,
\ee
\be
&&\delta^{M}(\left\{u,\bar u\right\})=h_3(u,u_M)\prod_{i=1}^{M-1}h(u,u_i)\non\\&&\qquad+
\sum_{i=1}^{M-1}\left\{h_1(u,u_i)h_3(u_i,u_M)\prod_{j\neq i}^{M-1}h(u_i,u_j)+
h_2(u,u_i)f_3(u_i,u_M)\prod_{j\neq i}^{M-1}f(u_i,u_j)\right.\non\\
&&\qquad\qquad-(q+q^{-1})Q^{-2}a(u)\sum_{k=2}^ir_{k-2}\left(-Q^{-2}d(u_i)f_3(u_i,u_M)\prod_{j=k,j\neq i}^{M-1}f(u_i,u_j)\right.\non\\
&&\qquad\qquad\qquad+\left.\left.
h_3(u_i,u_M)\prod_{j=k,j\neq i}^{M-1}h(u_i,u_j)\right)\right\}\,,
\ee
\be
&&r_i=-\left(\frac{1+Q^2}{Q^4}\right)^i\,,
\ee
\be
&& s_i=\frac{1}{Q^2}\left(\frac{1+Q^2}{Q^4}\right)^{i-1}\,.
\ee

The coefficients of the commutation relation (\ref{CB})
are given by,
{\allowdisplaybreaks
\be
&&x_1(u,v)=\frac{\omega(u^2) \left(Q^2 \omega(q v^2)+\omega(v^2)\right) \left(\omega(quv^{-1}) +
Q^{-2}\omega(uv^{-1})\right)}{Q^2 \omega(q u^2) \omega(q v^2) \omega(vu^{-1})}\,,\non\\
&&x_2(u,v)=\frac{\omega(u^2) \omega(quv^{-1}) \left(\omega(q u v)+Q^{-2}\omega(u v)\right)}{\omega(q u^2)
\omega(uv^{-1})\omega(q u v)}\,,\non\\
&&x_3(u,v)=
-\frac{\left(Q^2 \omega(q v^2)+\omega(v^2)\right)
   \left(\omega (q u v)+Q^2 \omega (u v)\right)}{Q^2 \omega(q v^2)
   \omega(q u v)}\,,\non\\
&&x_4(u,v)=
\frac{\omega(u^2) \left(\omega(quv^{-1})+Q^{-2}\omega(uv^{-1})\right)}{\omega(q u^2) \omega(vu^{-1})}\,,\non\\
&&x_5(u,v)=\frac{\omega (u v) \left(\omega(quv^{-1})+Q^2
\omega(uv^{-1})\right)}{\omega(uv^{-1}) \omega (q
   u v)}\,,\non\\
&&x_6(u,v)=-\frac{\omega (q u v)+Q^2 \omega (u v)}{\omega (q u v)}\,,\non\\
&&y_1(u,v)=\frac{\omega(u v)}{\omega(q u v)}\,,\non\\
&&y_2(u,v)=-\frac{\omega (u v) \left(Q^2 \omega(q v^2)+\omega
   (v^2)\right)}{Q^2 \omega(q v^2) \omega(q u v)}\,,\non\\
&&y_3(u,v)=-\frac{\omega(u v)}{\omega(q u v)}\,.
\ee
}
\section{Proof of Proposition \ref{prop1}}\label{sec:proofAonB}
In this appendix we prove Proposition \ref{prop1} by induction.
\subsection{Proof of (\ref{AonBstring})}
Let us consider the relation (\ref{AonBstring}).
Its validity for $M=1$ follows directly from the commutation relations (\ref{AB}) and (\ref{DB}). Let us suppose
that (\ref{AonBstring}) is valid for arbitrary $M$,
and compute the action
\be\label{step0}
\mA(u)B^{M+1}(\bar{\bar{u}})=\mA(u)B^M(\bar u)\mB(u_{M+1})
\ee
where $\bar{\bar{u}}=\left\{\bar u,u_{M+1}\right\}$ with $\#\bar u=M$. Using the induction hypothesis (\ref{AonBstring}) in (\ref{step0}) we obtain
\be\label{step1}
&&\mA(u)B^{M+1}(\bar{\bar{u}})
=B^M(\bar u)\mA(u)\mB(u_{M+1})\prod_{i=1}^Mf(u,u_i)+
\sum_{i=1}^M\hat F_i^M(\left\{u,\bar u\right\})\mB(u_{M+1})\non\\
&&\qquad+\sum_{i=2}^M\hat Z_i^M(\left\{u,\bar u\right\})\mB(u_{M+1})+
\sum_{i=0}^{M-1}r_i\bar B_i^M(\left\{u,\bar u\right\})\mB(u_{M+1})
\non\\
&&\qquad+\sum_{i=1}^{M-1}s_i\tilde B_i^M(\left\{u,\bar 
u\right\})\mB(u_{M+1})+\alpha^{M}(\left\{u,\bar u\right\})\tilde 
B_M^M(\left\{u,\bar u\right\})\mB(u_{M+1}) \,.
\ee
The next step consists of using the commutation relations
(\ref{AB}), (\ref{DB}), (\ref{BB}), (\ref{BB1}) and (\ref{BE}) in (\ref{step1}). Let us consider each term in the right-hand side of (\ref{step1}) separately. 
We have
\be\label{partial1}
&&B^M(\bar u)\underbrace{\mA(u)\mB(u_{M+1})}_{\eq(\ref{AB})}\prod_{i=1}^Mf(u,u_i)\non\\
&&\qquad=
B^{M+1}(\bar{\bar u})\mA(u)\prod_{i=1}^{M+1}f(u,u_i)
+\gamma_a^M\,B^M(\bar u)\mB(u)\mA(u_{M+1})
+\gamma_d^M\,B^M(\bar u)
\mB(u)\mD(u_{M+1})
\non\\&&\qquad\qquad
+
\gamma_e^M\,B^M(\bar u)\mB(u)\mE(u_{M+1})
+\gamma_{b_2}^M\,B^M(\bar u)\mB_1(u)\mB_2(u_{M+1})\,,
\ee
where
\be
&&\gamma_a^M=f_1(u,u_{M+1})\prod_{i=1}^Mf(u,u_i)\,,
\ee
\be
&&\gamma_d^M=f_2(u,u_{M+1})\prod_{i=1}^Mf(u,u_i)\,,
\ee
\be
&&\gamma_e^M=f_3(u,u_{M+1})\prod_{i=1}^Mf(u,u_i)\,,
\ee
\be
&&\gamma_{b_2}^M=-\prod_{i=1}^Mf(u,u_i)\,,
\ee
are auxiliary quantities introduced for convenience\footnote{Here, and in the auxiliary functions defined hereafter, we omit the functional dependency on the
rapidities in order to lighten the 
notation.}.
The next term is given by
{\allowdisplaybreaks
\be\label{partial2}
&&\sum_{i=1}^M\hat F_i^M(\left\{u,\bar u\right\})\mB(u_{M+1})
\non\\
&&\qquad=
\sum_{i=1}^M\Bigg\{
f_1(u,u_i)B_i^{M}(\left\{u,\bar u_i\right\})\underbrace{\mA(u_i)\mB(u_{M+1})}_{\eq(\ref{AB})}\prod_{j\neq i}^Mf(u_i,u_j)
\non\\&&\qquad\qquad+f_2(u,u_i)B_i^{M}(\left\{u,\bar u_i\right\})\underbrace{\mD(u_i)\mB(u_{M+1})}_{\eq(\ref{DB})}\prod_{j\neq i}^Mh(u_i,u_j)
\Bigg\}
\non\\
&&\qquad=
\sum_{i=1}^M\Bigg\{
f_1(u,u_i)B_i^{M+1}(\left\{u,\bar{\bar u}_i\right\})\mA(u_i)\prod_{j\neq i}^{M+1}f(u_i,u_j)
\non\\&&
\qquad\qquad+
f_2(u,u_i)B_i^{M+1}(\left\{u,\bar{\bar u}_i\right\})\mD(u_i)\prod_{j\neq i}^{M+1}h(u_i,u_j)
\Bigg\}
\non\\
&&\qquad+
\sum_{i=1}^M\Bigg\{
\bigg(f_1(u,u_i)f_1(u_i,u_{M+1})\prod_{j\neq i}^Mf(u_i,u_j)\non\\&&
\qquad\qquad\qquad+
f_2(u,u_i)h_2(u_i,u_{M+1})\prod_{j\neq i}^Mh(u_i,u_j)\bigg)\underbrace{B_i^{M}(\left\{u,\bar u_i\right\})\mB(u_i)}_{\eq(\ref{BB})}\mA(u_{M+1})
\non\\&&
\qquad\qquad+
\bigg(f_1(u,u_i)f_2(u_i,u_{M+1})\prod_{j\neq i}^Mf(u_i,u_j)
\non\\&&
\qquad\qquad\qquad+
f_2(u,u_i)h_1(u_i,u_{M+1})\prod_{j\neq i}^Mh(u_i,u_j)\bigg)\underbrace{B_i^{M}(\left\{u,\bar u_i\right\})\mB(u_i)}_{\eq(\ref{BB})}\mD(u_{M+1})
\non\\&&
\qquad\qquad+
\bigg(
f_1(u,u_i)f_3(u_i,u_{M+1})\prod_{j\neq i}^Mf(u_i,u_j)\non\\&&
\qquad\qquad+
f_2(u,u_i)h_3(u_i,u_{M+1})\prod_{j\neq i}^Mh(u_i,u_j)
\bigg)
\underbrace{B_i^{M}(\left\{u,\bar u_i\right\})\mB(u_i)}_{\eq(\ref{BB})}\mE(u_{M+1})
\non\\&&
\qquad\qquad-
\bigg(
f_1(u,u_i)\prod_{j\neq i}^Mf(u_i,u_j)\non\\&&
\qquad\qquad-
Q^{-2}f_2(u,u_i)a(u_i)\prod_{j\neq i}^Mh(u_i,u_j)
\bigg)
\underbrace{B_i^{M}(\left\{u,\bar u_i\right\})\mB_1(u_i)}_{\eq(\ref{BB1})}\mB_2(u_{M+1})
\non\\&&
\qquad\qquad-
\left(Q^{-2}f_2(u,u_i)\prod_{j\neq i}^Mh(u_i,u_j)\right)
\underbrace{B_i^{M}(\left\{u,\bar u_i\right\})\mE(u_i)}_{\eq(\ref{BE})}\mB(u_{M+1})
\Bigg\}
\non\\
&&\qquad=
\sum_{i=1}^{M+1}\hat F_i^{M+1}(\{u,\bar{\bar u}\})+
\theta_a^M\,B^M(\bar u)\mB(u)\mA(u_{M+1})
\non\\&&
\qquad\qquad+
\theta_d^M\,B^M(\bar u)\mB(u)\mD(u_{M+1})
+
\theta_e^M\,
B^M(\bar u)\mB(u)\mE(u_{M+1})
\non\\&&
\qquad\qquad+
\theta_{b_2}^M\,
B^M(\bar u)\mB_1(u)\mB_2(u_{M+1})
+
\theta_b^M\,
B^M(\bar u)\mE(u)\mB(u_{M+1})\,, 
\ee
}where we identified $B_i^{M}(\left\{u,\bar u_i\right\})\mB(u_{M+1})=B_i^{M+1}(\left\{u,\bar{\bar u}_i\right\})$
and
introduced the auxiliary functions
\be
&&\theta_a^M=-f_1(u,u_{M+1})\prod_{j=1}^{M}f(u_{M+1},u_j)+\sum_{i=1}^M\bigg(f_1(u,u_i)f_1(u_i,u_{M+1})\prod_{j\neq i}^Mf(u_i,u_j)\non\\&&
\qquad+
f_2(u,u_i)h_2(u_i,u_{M+1})\prod_{j\neq i}^Mh(u_i,u_j)\bigg)\,,
\ee
\be
&&\theta_d^M=-
f_2(u,u_{M+1})\prod_{j=1}^{M}h(u_{M+1},u_j)+\sum_{i=1}^M\bigg(f_1(u,u_i)f_2(u_i,u_{M+1})\prod_{j\neq i}^Mf(u_i,u_j)
\non\\&&
\qquad+
f_2(u,u_i)h_1(u_i,u_{M+1})\prod_{j\neq i}^Mh(u_i,u_j)\bigg)\,,
\ee
\be
&&\theta_e^M=\sum_{i=1}^M\bigg(
f_1(u,u_i)f_3(u_i,u_{M+1})\prod_{j\neq i}^Mf(u_i,u_j)+
f_2(u,u_i)h_3(u_i,u_{M+1})\prod_{j\neq i}^Mh(u_i,u_j)\bigg)\,,\non\\
\ee
\be
\theta_{b_2}^M=-\sum_{i=1}^M\bigg(
f_1(u,u_i)\prod_{j\neq i}^Mf(u_i,u_j)-
Q^{-2}f_2(u,u_i)a(u_i)\prod_{j\neq i}^Mh(u_i,u_j)\bigg)\,,
\ee
\be
\theta_b^M=-\sum_{i=1}^MQ^{-2}f_2(u,u_i)\prod_{j\neq i}^Mh(u_i,u_j)\,.
\ee
We proceed in a similar way for the next term, namely,
\be\label{partial3}
&&\sum_{i=2}^M\hat Z_i^M(\left\{u,\bar u\right\})\mB(u_{M+1})
\non\\
&&\qquad=
\sum_{i=2}^M(q+q^{-1})\sum_{k=2}^ir_{k-2}\times
\Bigg\{
-Q^{-2}d(u_i)B_i^{M}(\left\{u,\bar u_i\right\})\underbrace{\mA(u_i)\mB(u_{M+1})}_{\eq(\ref{AB})}\prod_{j=k,j\neq i}^Mf(u_i,u_j)\non\\
&&\qquad\qquad+
B_i^{M}(\left\{u,\bar u_i\right\})\underbrace{\mD(u_i)\mB(u_{M+1})}_{\eq(\ref{DB})}\prod_{j=k,j\neq i}^Mh(u_i,u_j)\Bigg\}
\non\\
&&\qquad=
\sum_{i=2}^{M+1}\hat Z_i^{M+1}(\left\{u,\bar{\bar u}\right\})+\tau_a^MB^M(\bar u)\mB(u)\mA(u_{M+1})
\non\\
&&\qquad\qquad+
\tau_d^MB^M(\bar u)\mB(u)\mD(u_{M+1})
+\tau_e^MB^M(\bar u)\mB(u)\mE(u_{M+1})
\non\\
&&\qquad\qquad
+\tau_{b_2}^MB^M(\bar u)\mB_1(u)\mB_2(u_{M+1})
+\tau_b^MB^M(\bar u)\mE(u)\mB(u_{M+1}) \,,
\ee
where we used the relations (\ref{BB}), (\ref{BB1}) and (\ref{BE}) to
rewrite, respectively, the terms $B_i^{M}(\left\{u,\bar u_i\right\})\mB(u_i)$, $B_i^{M}(\left\{u,\bar u_i\right\})\mB_1(u_i)$ and
$B_i^{M}(\left\{u,\bar u_i\right\})\mE(u_i)$. The auxiliary functions $\tau_i$ are given by
\be
&&\tau_a^M=(q+q^{-1})Q^{-2}d(u_{M+1})\sum_{k=2}^{M+1}r_{k-2}\prod_{j=k}^{M}f(u_{M+1},u_j)
\non\\&&\qquad-
\sum_{i=2}^M
(q+q^{-1})\sum_{k=2}^ir_{k-2}\bigg(
Q^{-2}d(u_i)f_1(u_i,u_{M+1})\prod_{j=k,j\neq i}^Mf(u_i,u_j)
\non\\&&\qquad\qquad-h_2(u_i,u_{M+1})\prod_{j=k,j\neq i}^Mh(u_i,u_j)\bigg)\,,
\ee
\be
&&\tau_d^M=-
(q+q^{-1})\sum_{k=2}^{M+1}r_{k-2}
\prod_{j=k}^{M}h(u_{M+1},u_j)
\non\\&&\qquad-\sum_{i=2}^M(q+q^{-1})\sum_{k=2}^ir_{k-2}\bigg(Q^{-2}d(u_i)f_2(u_i,u_{M+1})\prod_{j=k,j\neq i}^Mf(u_i,u_j)
\non\\&&\qquad\qquad-
h_1(u_i,u_{M+1})\prod_{j=k,j\neq i}^Mh(u_i,u_j)\bigg)\,,
\ee
\be
&&\tau_e^M=-\sum_{i=2}^M(q+q^{-1})\sum_{k=2}^ir_{k-2}\bigg(Q^{-2}d(u_i)f_3(u_i,u_{M+1})\prod_{j=k,j\neq i}^Mf(u_i,u_j)
\non\\&&\qquad-
h_3(u_i,u_{M+1})\prod_{j=k,j\neq i}^Mh(u_i,u_j)\bigg)\,,
\ee
\be
&&\tau_{b_2}^M=\sum_{i=2}^M(q+q^{-1})\sum_{k=2}^ir_{k-2}Q^{-2}\bigg(d(u_i)\prod_{j=k,j\neq i}^Mf(u_i,u_j)
\non\\&&\qquad+
a(u_i)\prod_{j=k,j\neq i}^Mh(u_i,u_j)\bigg)\,,
\ee
\be
&&\tau_b^M=-\sum_{i=2}^M
(q+q^{-1})\sum_{k=2}^ir_{k-2}Q^{-2}\prod_{j=k,j\neq i}^Mh(u_i,u_j)\,.
\ee
The last terms are
\be\label{partial4}
&&\sum_{i=0}^{M-1}r_i\bar B_i^M(\left\{u,\bar u\right\})\mB(u_{M+1})
\non\\&&\qquad=
\sum_{i=0}^{M}r_i\bar B_i^{M+1}(\left\{u,\bar{\bar u}\right\})-r_M\underbrace{\bar B_M^{M+1}(\left\{u,\bar{\bar u}\right\})}_{\eq(\ref{BB1})}
\non\\&&\qquad=
\sum_{i=0}^{M}r_i\bar B_i^{M+1}(\left\{u,\bar{\bar u}\right\})-r_MB^M(\bar u)\mB_1(u)\mB_2(u_{M+1})\,,
\ee
\be\label{partial5}
&&\sum_{i=1}^{M-1}s_i\tilde B_i^M(\left\{u,\bar u\right\})\mB(u_{M+1})
\non\\&&\qquad=\sum_{i=1}^{M}s_i\tilde B_i^{M+1}(\left\{u,\bar{\bar u}\right\})-s_M\underbrace{\tilde B_M^{M+1}(\left\{u,\bar{\bar u}\right\})}_{\eq(\ref{BE})}
\non\\&&\qquad=\sum_{i=1}^{M}s_i\tilde B_i^{M+1}(\left\{u,\bar{\bar 
u}\right\})-s_MB^M(\bar u)\mE(u)\mB(u_{M+1}) \,,
\ee
\be\label{partial6}
&&\alpha^{M}(\left\{u,\bar u\right\})\underbrace{\tilde B_M^M(\left\{u,\bar u\right\})\mB(u_{M+1})}_{\eq(\ref{BE})}=
\alpha^{M}(\left\{u,\bar u\right\})B^M(\bar u)\mE(u)\mB(u_{M+1}) \,.
\ee
We observe the following identities
\be
&&\label{idaa}\gamma_a^M+\theta_a^M+\tau_a^M=0\,,\\
&&\label{iddd}\gamma_d^M+\theta_d^M+\tau_d^M=0\,,\\
&&\label{idee}\gamma_e^M+\theta_e^M+\tau_e^M=\alpha^{M+1}(\{u,\bar{\bar u}\})\,,\\
&&\label{idb2}\gamma_{b_2}^M+\theta_{b_2}^M+\tau_{b_2}^M=r_M\,,\\
&&\label{idbb}\theta_b^M+\tau_b^M+\alpha^{M}(\{u,\bar u\})=s_M\,,
\ee
which are typical in algebraic Bethe ansatz
analyses, see \textit{e.g.} equations
(A.8) and (A.9) in \cite{Avan:2015ada}.  As an example here, let us show the validity of the simplest relation (\ref{idb2}), using analytical
arguments.
We start by calculating the residues of the left-hand side of (\ref{idb2}); we note that
\be
\textrm{Res}\left(\gamma_{b_2}^M,u=u_{\textrm{pole}}\right)=-\textrm{Res}\left(\theta_{b_2}^M,u=u_{\textrm{pole}}\right)
\ee
and
\be
\textrm{Res}\left(\tau_{b_2}^M,u=u_{\textrm{pole}}\right)=0
\ee
where $u_{\textrm{pole}}=u_k, -u_k, q^{-1}u_k^{-1},-q^{-1}u_k^{-1}$ for $k=1,\dots,M$. This shows that the residue of the left-hand side of the functional relation
is zero; therefore it is holomorphic on the entire complex plane, and thus equals a constant. The constant can be determined by taking the limit:
\be
&&\lim_{u\rightarrow\infty}\gamma_{b_2}^M+\theta_{b_2}^M+\tau_{b_2}^M=r_M\,.
\ee
The other functional relations (except for 
(\ref{idee}), which is trivial since it is basically the definition 
of $\alpha^M$) can be analyzed in the same way.

Finally, using the results (\ref{partial1}), (\ref{partial2}), (\ref{partial3}), (\ref{partial4}), (\ref{partial5}) and (\ref{partial6}) in (\ref{step1}),
the functional identities (\ref{idaa}) - (\ref{idbb}),
as well as noticing that $B^M(\bar u)\mB(u)\mE(u_{M+1})=\tilde{B}_{M+1}^{M+1}(\{u,\bar{\bar u}\})$,
we obtain
\be
&&\mA(u)B^{M+1}(\bar{\bar{u}})=B^{M+1}(\bar{\bar u})\mA(u)\prod_{i=1}^{M+1}f(u,u_i)+
\sum_{i=1}^{M+1}\hat F_i^{M+1}(\left\{u,\bar{\bar{u}}\right\})+\sum_{i=2}^{M+1}\hat Z_i^{M+1}(\left\{u,\bar{\bar{u}}\right\})\non\\
&&\qquad+
\sum_{i=0}^{M}r_i\bar B_i^{M+1}(\left\{u,\bar{\bar{u}}\right\})
+\sum_{i=1}^{M}s_i\tilde B_i^{M+1}(\left\{u,\bar{\bar{u}}\right\})+
\alpha^{M+1}(\left\{u,\bar{\bar{u}}\right\})\tilde 
B_{M+1}^{M+1}(\left\{u,\bar{\bar{u}}\right\}) \,,
\ee
which ends the proof.

\subsection{Proof of (\ref{DonBstring})}

Let us now consider the relation (\ref{DonBstring}).
Its validity for $M=1$ follows directly from the commutation relations (\ref{AB}) and (\ref{DB}). Let us suppose
that (\ref{DonBstring}) is valid for arbitrary $M$,
and compute the action
\be\label{step0a}
\mD(u)B^{M+1}(\bar{\bar{u}})=\mD(u)B^M(\bar u)\mB(u_{M+1}) \,,
\ee
where $\bar{\bar{u}}=\left\{\bar u,u_{M+1}\right\}$ with $\#\bar u=M$. Using the induction hypothesis (\ref{DonBstring}) in (\ref{step0a}) we obtain
\be\label{step1a}
&&\mD(u)B^{M+1}(\bar{\bar{u}})=
B^M(\bar u)\mD(u)\mB(u_{M+1})\prod_{i=1}^Mh(u,u_i)
+\sum_{i=1}^M\hat G_i^M(\left\{u,\bar u\right\})\mB(u_{M+1})\non\\
&&\qquad-Q^{-2}a(u)\sum_{i=2}^M\hat Z_i^M(\left\{u,\bar u\right\})\mB(u_{M+1})
-	
Q^{-2}a(u)\sum_{i=0}^{M-1}r_i\bar B_i^M(\left\{u,\bar u\right\})\mB(u_{M+1})\non\\
&&\qquad
-Q^{-2}a(u)\sum_{i=1}^{M-1}s_i\tilde B_i^M(\left\{u,\bar u\right\})\mB(u_{M+1})+
\delta^{M}(\left\{u,\bar u\right\})\tilde B_M^M(\left\{u,\bar u\right\})\mB(u_{M+1})\non\\&&
\qquad-Q^{-2}\mE(u)B^M(\bar u)\mB(u_{M+1}) \,.
\ee
The next step consists of using the commutation relations
(\ref{AB}), (\ref{DB}), (\ref{BB}), (\ref{BB1}) and (\ref{BE}) in 
(\ref{step1a}). Most of the terms have already been
computed in the previous subsection; thus, we need to compute here only the first two and the last two terms in the right-hand side of (\ref{step1a}).
We have
\be\label{partial1a}
\lefteqn{B^M(\bar u)\underbrace{\mD(u)\mB(u_{M+1})}_{\eq(\ref{DB})}\prod_{i=1}^Mh(u,u_i)}
\non\\
&&=
B^{M+1}(\bar{\bar u})\mD(u)\prod_{i=1}^{M+1}h(u,u_i)
+\bar\gamma_d^M B^M(\bar u)\mB(u)\mD(u_{M+1})
+\bar\gamma_a^MB^M(\bar u)\mB(u)\mA(u_{M+1})\non\\&&
\qquad+
\bar\gamma_e^M B^M(\bar u)\mB(u)\mE(u_{M+1})
+\bar\gamma_{b_2}^M B^M(\bar u)\mB_1(u)\mB_2(u_{M+1})
+\bar\gamma_{b}^MB^M(\bar u)\mE(u)\mB(u_{M+1})\,, \non\\
\ee
where
\be
&&\bar\gamma_a^M=h_2(u,u_{M+1})\prod_{i=1}^Mh(u,u_i)\,,
\ee
\be
&&\bar\gamma_d^M=h_1(u,u_{M+1})\prod_{i=1}^Mh(u,u_i)\,,
\ee
\be
&&\bar\gamma_e^M=h_3(u,u_{M+1})\prod_{i=1}^Mh(u,u_i)\,,
\ee
\be
&&\bar\gamma_{b_2}^M=Q^{-2}a(u)\prod_{i=1}^Mh(u,u_i)\,,
\ee
\be
&&\bar\gamma_{b}^M=-Q^{-2}\prod_{i=1}^Mh(u,u_i)\,.
\ee
The next term is computed in a similar way as (\ref{partial2}). The result is given by
\be\label{partial2a}
&&\sum_{i=1}^M\hat G_i^M(\left\{u,\bar u\right\})\mB(u_{M+1})=
\sum_{i=1}^{M+1}\hat G_i^{M+1}(\{u,\bar{\bar u}\})+
\bar\theta_a^M\,B^M(\bar u)\mB(u)\mA(u_{M+1})
\non\\&&
\qquad\qquad+
\bar\theta_d^M\,B^M(\bar u)\mB(u)\mD(u_{M+1})
+
\bar\theta_e^M\,
B^M(\bar u)\mB(u)\mE(u_{M+1})
\non\\&&
\qquad\qquad+
\bar\theta_{b_2}^M\,
B^M(\bar u)\mB_1(u)\mB_2(u_{M+1})
+
\bar\theta_b^M\,
B^M(\bar u)\mE(u)\mB(u_{M+1}) \,,
\ee
where we introduced the auxiliary functions
\be
&&\bar\theta_a^M=-h_2(u,u_{M+1})\prod_{j=1}^{M}f(u_{M+1},u_j)+\sum_{i=1}^M\bigg(h_2(u,u_i)f_1(u_i,u_{M+1})\prod_{j\neq i}^Mf(u_i,u_j)\non\\&&
\qquad+
h_1(u,u_i)h_2(u_i,u_{M+1})\prod_{j\neq i}^Mh(u_i,u_j)\bigg)\,,
\ee
\be
&&\bar\theta_d^M=-
h_1(u,u_{M+1})\prod_{j=1}^{M}h(u_{M+1},u_j)+\sum_{i=1}^M\bigg(h_2(u,u_i)f_2(u_i,u_{M+1})\prod_{j\neq i}^Mf(u_i,u_j)
\non\\&&
\qquad+
h_1(u,u_i)h_1(u_i,u_{M+1})\prod_{j\neq i}^Mh(u_i,u_j)\bigg)\,,
\ee
\be
&&\bar\theta_e^M=\sum_{i=1}^M\bigg(
h_2(u,u_i)f_3(u_i,u_{M+1})\prod_{j\neq i}^Mf(u_i,u_j)+
h_1(u,u_i)h_3(u_i,u_{M+1})\prod_{j\neq i}^Mh(u_i,u_j)\bigg)\,,\non\\
\ee
\be
\bar\theta_{b_2}^M=-\sum_{i=1}^M\bigg(
h_2(u,u_i)\prod_{j\neq i}^Mf(u_i,u_j)-
Q^{-2}h_1(u,u_i)a(u_i)\prod_{j\neq i}^Mh(u_i,u_j)\bigg)\,,
\ee
\be
\bar\theta_b^M=-\sum_{i=1}^MQ^{-2}h_1(u,u_i)\prod_{j\neq i}^Mh(u_i,u_j)\,.
\ee
The last terms are easily evaluated
\be\label{partial6a}
&&\delta^{M}(\left\{u,\bar u\right\})\underbrace{\tilde B_M^M(\left\{u,\bar u\right\})\mB(u_{M+1})}_{\eq(\ref{BE})}=
\delta^{M}(\left\{u,\bar u\right\})B^M(\bar u)\mE(u)\mB(u_{M+1}) \,,
\ee
\be\label{partiala}
-Q^{-2}\mE(u)B^M(\bar u)\mB(u_{M+1})=-Q^{-2}\mE(u)B^{M+1}(\bar{\bar u})\,.
\ee
We now observe identities that are analogous to (\ref{idaa}) - (\ref{idbb}), namely,
\be
&&\label{idaa1}\bar\gamma_a^M+\bar\theta_a^M-Q^{-2}a(u)\tau_a^M=0\,,\\
&&\label{iddd1}\bar\gamma_d^M+\bar\theta_d^M-Q^{-2}a(u)\tau_d^M=0\,,\\
&&\label{idee1}\bar\gamma_e^M+\bar\theta_e^M-Q^{-2}a(u)\tau_e^M=\delta^{M+1}(\{u,\bar{\bar u}\})\,,\\
&&\label{idb21}\bar\gamma_{b_2}^M+\bar\theta_{b_2}^M-Q^{-2}a(u)\tau_{b_2}^M=-Q^{-2}a(u)r_M\,,\\
&&\label{idbb1}\bar\gamma_b^M+\bar\theta_b^M-Q^{-2}a(u)\tau_b^M+\delta^{M}(\{u,\bar u\})=-Q^{-2}a(u)s_M\,.
\ee
Finally, using the results (\ref{partial1a}), (\ref{partial2a}), (\ref{partial6a}), (\ref{partiala}),  (\ref{partial3}), (\ref{partial4}) and (\ref{partial5}) in (\ref{step1a}),
the functional identities (\ref{idaa1}) - (\ref{idbb1}), as well as noticing that $B^M(\bar u)\mB(u)\mE(u_{M+1})=\tilde{B}_{M+1}^{M+1}(\{u,\bar{\bar u}\})$,
we obtain
\be
&&\mD(u)B^{M+1}(\bar{\bar u})=B^{M+1}(\bar{\bar u})\mD(u)\prod_{i=1}^{M+1}h(u,u_i)+\sum_{i=1}^{M+1}\hat G_i^{M+1}(\left\{u,\bar{\bar u}\right\})
\non\\
&&\qquad-Q^{-2}a(u)\sum_{i=2}^{M+1}\hat Z_i^{M+1}(\left\{u,\bar{\bar u}\right\})
-	
Q^{-2}a(u)\sum_{i=0}^{M}r_i\bar B_i^{M+1}(\left\{u,\bar{\bar u}\right\})
\non\\
&&\qquad
-Q^{-2}a(u)\sum_{i=1}^{M}s_i\tilde B_i^{M+1}(\left\{u,\bar{\bar u}\right\})+\delta^{M+1}(\left\{u,\bar u\right\})\tilde B_{M+1}^{M+1}(\left\{u,\bar{\bar u}\right\})
-Q^{-2}\mE(u)B^{M+1}(\bar{\bar u})\,,\non\\
\ee
which concludes the proof.

%\bibliographystyle{utphys}
%\bibliography{refs}

\begin{thebibliography}{10}

\bibitem{Faddeev:1996iy}
L.~D. Faddeev, ``{How algebraic Bethe ansatz works for integrable models},'' in
  {\em Sym\'etries Quantiques (Les Houches Summer School Proceedings vol 64)},
  A.~Connes, K.~Gawedzki, and J.~Zinn-Justin, eds., pp.~149--219.
\newblock North Holland, 1998.
\newblock
\href{http://arxiv.org/abs/hep-th/9605187}{{\ttfamily arXiv:hep-th/9605187
  [hep-th]}}.
\newblock
%%CITATION = HEP-TH/9605187;%%.

\bibitem{Sklyanin:1988yz}
E.~K. Sklyanin, ``{Boundary Conditions for Integrable Quantum Systems},''
\href{http://dx.doi.org/10.1088/0305-4470/21/10/015}{{\em J. Phys.} {\bfseries
  A21} (1988) 2375--289}.
%%CITATION = JPAGA,A21,2375;%%.

\bibitem{KBI}
V.~E. {Korepin}, N.~M. {Bogoliubov}, and A.~G. {Izergin}, {\em {Quantum Inverse
  Scattering Method and Correlation Functions}}.
\newblock {Cambridge University Press}, 1997.

\bibitem{Temperley:1971iq}
H.~N.~V. Temperley and E.~H. Lieb, ``{Relations between the 'percolation' and
  'colouring' problem and other graph-theoretical problems associated with
  regular planar lattices: some exact results for the 'percolation' problem},''
\href{http://dx.doi.org/10.1098/rspa.1971.0067}{{\em Proc. Roy. Soc. Lond.}
  {\bfseries A322} (1971) 251--280}.
%%CITATION = PRSLA,A322,251;%%.

\bibitem{Batchelor:1989uk}
M.~T. Batchelor, L.~Mezincescu, R.~I. Nepomechie, and V.~Rittenberg, ``{$q$
  Deformations of the O(3) Symmetric Spin 1 Heisenberg Chain},''
\href{http://dx.doi.org/10.1088/0305-4470/23/4/003}{{\em J. Phys.} {\bfseries
  A23} (1990) L141}.
%%CITATION = JPAGA,A23,L141;%%.

\bibitem{Jones:1990hq}
V.~F.~R. Jones, ``{Baxterization},''
\href{http://dx.doi.org/10.1142/S0217751X91001027}{{\em Int. J. Mod. Phys.}
  {\bfseries A6} (1991) 2035--2043}.
%%CITATION = IMPAE,A6,2035;%%.

\bibitem{Nepomechie:2016ejv}
R.~I. Nepomechie and R.~A. Pimenta, ``{Universal Bethe ansatz solution for the
  Temperley–Lieb spin chain},''
  \href{http://dx.doi.org/10.1016/j.nuclphysb.2016.04.045}{{\em Nucl. Phys.}
  {\bfseries B910} (2016) 910--928},
\href{http://arxiv.org/abs/1601.04378}{{\ttfamily arXiv:1601.04378 [math-ph]}}.
%%CITATION = ARXIV:1601.04378;%%.

\bibitem{Parkinson1}
J.~B. Parkinson, ``{On the integrability of the S=1 quantum spin chain with
  pure biquadratic exchange},'' {\em J. Phys. C: Solid State Phys.} {\bfseries
  20} no.~36, (1987) L1029.

\bibitem{Parkinson2}
J.~B. Parkinson, ``{The S=1 quantum spin chain with pure biquadratic
  exchange},'' {\em J. Phys. C: Solid State Phys.} {\bfseries 21} (1988) 3793.

\bibitem{Barber:1989zz}
M.~N. Barber and M.~T. Batchelor, ``{Spectrum of the biquadratic spin-1
  antiferromagnetic chain},''
\href{http://dx.doi.org/10.1103/PhysRevB.40.4621}{{\em Phys. Rev.} {\bfseries
  B40} (1989) 4621--4626}.
%%CITATION = PHRVA,B40,4621;%%.

\bibitem{Kluemper1}
A.~{Kl{\"u}mper}, ``{New results for q-state vertex models and the pure
  biquadratic spin-1 Hamiltonian},''
  \href{http://dx.doi.org/10.1209/0295-5075/9/8/013}{{\em EPL (Europhysics
  Letters)} {\bfseries 9} (1989) 815}.

\bibitem{Kluemper2}
A.~{Kl{\"u}mper}, ``{The spectra of q-state vertex models and related
  antiferromagnetic quantum spin chains},''
  \href{http://dx.doi.org/10.1088/0305-4470/23/5/023}{{\em J. Phys. A: Math.
  Gen.} {\bfseries 23} (1990) 809--823}.

\bibitem{Alcaraz:1992uq}
F.~C. Alcaraz and A.~L. Malvezzi, ``{On the critical behavior of the
  anisotropic biquadratic spin 1 chain},''
{\em J. Phys.} {\bfseries A25} (1992) 4535--4546.
%%CITATION = JPAGA,A25,4535;%%.

\bibitem{Koberle:1993in}
R.~Koberle and A.~Lima-Santos, ``{Exact solution of the deformed biquadratic
  spin 1 chain},'' \href{http://dx.doi.org/10.1088/0305-4470/27/16/009}{{\em J.
  Phys.} {\bfseries 27} (1994) 5409},
\href{http://arxiv.org/abs/hep-th/9302140}{{\ttfamily arXiv:hep-th/9302140
  [hep-th]}}.
%%CITATION = HEP-TH/9302140;%%.

\bibitem{Kulish}
P.~Kulish, ``{On spin systems related to the Temperley-Lieb algebra},'' {\em J.
  Phys. A: Math. Gen.} {\bfseries 36} (2003) L489.

\bibitem{Aufgebauer:2010gg}
B.~Aufgebauer and A.~Kl{\"u}mper, ``{Quantum spin chains of Temperley-Lieb
  type: periodic boundary conditions, spectral multiplicities and finite
  temperature},''
  \href{http://dx.doi.org/10.1088/1742-5468/2010/05/P05018}{{\em J. Stat.
  Mech.} {\bfseries 1005} (2010) P05018},
\href{http://arxiv.org/abs/1003.1932}{{\ttfamily arXiv:1003.1932
  [cond-mat.stat-mech]}}.
%%CITATION = ARXIV:1003.1932;%%.

\bibitem{Tarasov}
V.~Tarasov, ``{Algebraic Bethe ansatz for the Izergin-Korepin R matrix},''
  \href{http://dx.doi.org/10.1007/BF01028578}{{\em Theor. Math. Phys.}
  {\bfseries 76} no.~2, (1988) 793--803}.

\bibitem{Kulish:1983rd}
P.~P. Kulish and N.~{\relax Yu}. Reshetikhin, ``{Diagonalisation of GL(n)
  invariant transfer matrices and quantum N-wave system (Lee model)},''
\href{http://dx.doi.org/10.1088/0305-4470/16/16/001}{{\em J. Phys.} {\bfseries
  A16} (1983) L591--L596}.
%%CITATION = JPAGA,A16,L591;%%.

\bibitem{Fan1997409}
H.~Fan, ``{Bethe ansatz for the Izergin-Korepin model},'' {\em Nucl. Phys.}
  {\bfseries B488} no.~1-2, (1997) 409 -- 425.

\bibitem{FOERSTER1993512}
A.~Foerster and M.~Karowski, ``{The supersymmetric t-J model with quantum group
  invariance},''
  \href{http://dx.doi.org/http://dx.doi.org/10.1016/0550-3213(93)90377-2}{{\em
  Nucl. Phys.} {\bfseries B408} no.~3, (1993) 512 -- 534}.

\bibitem{Kulish:1991np}
P.~P. Kulish and E.~K. Sklyanin, ``{The general U(q)(sl(2)) invariant XXZ
  integrable quantum spin chain},''
{\em J. Phys.} {\bfseries A24} (1991) L435--L439.
%%CITATION = JPAGA,A24,L435;%%.

\bibitem{Gainutdinov:2015vba}
A.~M. Gainutdinov, W.~Hao, R.~I. Nepomechie, and A.~J. Sommese, ``{Counting
  solutions of the Bethe equations of the quantum group invariant open XXZ
  chain at roots of unity},''
  \href{http://dx.doi.org/10.1088/1751-8113/48/49/494003}{{\em J. Phys.}
  {\bfseries A48} no.~49, (2015) 494003},
\href{http://arxiv.org/abs/1505.02104}{{\ttfamily arXiv:1505.02104 [math-ph]}}.
%%CITATION = ARXIV:1505.02104;%%.

\bibitem{Sla89}
N.~A. Slavnov, ``{Calculation of scalar products of wave functions and form
  factors in the framework of the algebraic Bethe ansatz},''
  \href{http://dx.doi.org/10.1007/BF01016531}{{\em Theoret. and Math. Phys.}
  {\bfseries 79} (1989) 502--508}.

\bibitem{Gaudin}
M.~Gaudin, {\em La fonction d'onde de Bethe}.
\newblock Masson, 1933.
\newblock English translation by J.-S. Caux, The Bethe Wavefunction, Cambridge
  University Press, 2014.

\bibitem{PhysRevD.23.417}
M.~Gaudin, B.~M. McCoy, and T.~T. Wu, ``{Normalization sum for the Bethe's
  hypothesis wave functions of the Heisenberg-Ising chain},''
  \href{http://dx.doi.org/10.1103/PhysRevD.23.417}{{\em Phys. Rev. D}
  {\bfseries 23} (1981) 417--419}.

\bibitem{Korepin:1982gg}
V.~E. Korepin, ``{Calculation of norms of Bethe wave functions },''
\href{http://dx.doi.org/10.1007/BF01212176}{{\em Commun. Math. Phys.}
  {\bfseries 86} (1982) 391--418}.
%%CITATION = CMPHA,86,391;%%.

\bibitem{Kitanine:2007bi}
N.~Kitanine, K.~K. Kozlowski, J.~M. Maillet, G.~Niccoli, N.~A. Slavnov, and
  V.~Terras, ``{Correlation functions of the open XXZ chain I},''
  \href{http://dx.doi.org/10.1088/1742-5468/2007/10/P10009}{{\em J. Stat.
  Mech.} {\bfseries 0710} (2007) P10009},
\href{http://arxiv.org/abs/0707.1995}{{\ttfamily arXiv:0707.1995 [hep-th]}}.
%%CITATION = ARXIV:0707.1995;%%.

\bibitem{Wang2002633}
Y.-S. Wang, ``The scalar products and the norm of {B}ethe eigenstates for the
  boundary {XXX} {H}eisenberg spin-1/2 finite chain,''
  \href{http://dx.doi.org/http://dx.doi.org/10.1016/S0550-3213(01)00610-1}{{\em
  Nucl.Phys.} {\bfseries B622} no.~3, (2002) 633 -- 649}.

\bibitem{Finch:2014nxa}
P.~E. Finch, M.~Flohr, and H.~Frahm, ``{Integrable anyon chains: from fusion
  rules to face models to effective field theories},''
  \href{http://dx.doi.org/10.1016/j.nuclphysb.2014.10.017}{{\em Nucl. Phys.}
  {\bfseries B889} (2014) 299--332},
\href{http://arxiv.org/abs/1408.1282}{{\ttfamily arXiv:1408.1282
  [cond-mat.str-el]}}.
%%CITATION = ARXIV:1408.1282;%%.

\bibitem{Finch:2015}
P.~E. Finch, R.~Weston, and P.~Zinn-Justin, ``{Theta function solution of the
  qKZB equation for a face model},''
  \href{http://arxiv.org/abs/1509.04594}{{\ttfamily arXiv:1509.04594
  [math-ph]}}.

\bibitem{LimaSantos:2010nw}
A.~Lima-Santos, ``{On the ${\cal{U}}_{q}[sl(2)]$ Temperley-Lieb reflection
  matrices},'' \href{http://dx.doi.org/10.1088/1742-5468/2011/01/P01009}{{\em
  J. Stat. Mech.} {\bfseries 1101} (2011) P01009},
\href{http://arxiv.org/abs/1011.2891}{{\ttfamily arXiv:1011.2891 [nlin.SI]}}.
%%CITATION = ARXIV:1011.2891;%%.

\bibitem{Avan:2010mh}
J.~Avan, P.~Kulish, and G.~Rollet, ``{Reflection $K$-matrices related to
  Temperley-Lieb $R$-matrices},''
  \href{http://dx.doi.org/10.1007/s11232-011-0130-y}{{\em Theor. Math. Phys.}
  {\bfseries 169} (2011) 1530},
\href{http://arxiv.org/abs/1012.3012}{{\ttfamily arXiv:1012.3012 [nlin.SI]}}.
%%CITATION = ARXIV:1012.3012;%%.

\bibitem{Belliard:2014fsa}
S.~Belliard, ``{Modified algebraic Bethe ansatz for XXZ chain on the segment
  – I: Triangular cases},''
  \href{http://dx.doi.org/10.1016/j.nuclphysb.2015.01.003}{{\em Nucl. Phys.}
  {\bfseries B892} (2015) 1--20},
\href{http://arxiv.org/abs/1408.4840}{{\ttfamily arXiv:1408.4840 [math-ph]}}.
%%CITATION = ARXIV:1408.4840;%%.

\bibitem{Belliard:2014rna}
S.~Belliard and R.~A. Pimenta, ``{Modified algebraic Bethe ansatz for XXZ chain
  on the segment – II – general cases},''
  \href{http://dx.doi.org/10.1016/j.nuclphysb.2015.03.016}{{\em Nucl. Phys.}
  {\bfseries B894} (2015) 527--552},
\href{http://arxiv.org/abs/1412.7511}{{\ttfamily arXiv:1412.7511 [math-ph]}}.
%%CITATION = ARXIV:1412.7511;%%.

\bibitem{Avan:2015ada}
J.~Avan, S.~Belliard, N.~Grosjean, and R.~A. Pimenta, ``{Modified algebraic
  Bethe ansatz for XXZ chain on the segment – III – Proof},''
\href{http://dx.doi.org/10.1016/j.nuclphysb.2015.08.006}{{\em Nucl. Phys.}
  {\bfseries B899} (2015) 229--246}.
%%CITATION = NUPHA,B899,229;%%.

\bibitem{Wang2015}
Y.~Wang, W.-L. Yang, J.~Cao, and K.~Shi, {\em Off-Diagonal Bethe Ansatz for
  Exactly Solvable Models}.
\newblock Springer, 2015.

\end{thebibliography}

\providecommand{\href}[2]{#2}\begingroup\raggedright\endgroup

\end{document}